\documentclass[journal, onecolumn, draftcls]{IEEEtran}
\usepackage{amsmath,amsfonts}
\usepackage{algorithmic}
\usepackage{array}
\usepackage[caption=false,font=footnotesize,labelfont=rm,textfont=rm]{subfig}

\usepackage{textcomp}
\usepackage{stfloats}
\usepackage{url}
\usepackage{verbatim}
\usepackage{graphicx}
\usepackage{balance}

\usepackage{amsmath}
\usepackage{makecell}
\usepackage{caption}
\usepackage{subcaption}

\usepackage{longtable}
\usepackage{lipsum}
\usepackage{tabulary}
\usepackage{graphicx}
\usepackage{subfiles}

\usepackage[open,openlevel=3,atend,numbered]{bookmark}

\usepackage[table,dvipsnames]{xcolor}
\usepackage{listings}

\usepackage{enumitem}
\usepackage{array}
\usepackage{multirow}
\usepackage{makecell}
\usepackage[export]{adjustbox}
\usepackage{amsthm,amsmath,amssymb}
\usepackage{mathrsfs}

\usepackage{orcidlink}
\newcounter{magicrownumbers}

\def\BibTeX{{\rm B\kern-.05em{\sc i\kern-.025em b}\kern-.08em
    T\kern-.1667em\lower.7ex\hbox{E}\kern-.125emX}}
    	\clearpage
    	\makeatletter 

\usepackage[backend=bibtex,style=ieee,sorting=none,url = false, doi = false,isbn = false]{biblatex}
\AtEveryBibitem{\clearfield{language}\clearlist{language}\clearfield{langid}\clearlist{langid}} 

\theoremstyle{plain}
\newtheorem{theorem}{Theorem}
\newtheorem{lemma}{Lemma}
\newtheorem{corollary}{Corollary}
\newtheorem{proposition}{Proposition}

\theoremstyle{definition}

\theoremstyle{remark}
\newtheorem{remark}{Remark}
\newtheorem*{Proof}{Proof}

\begin{document}
\title{New Upper Bounds for Noisy Permutation Channels}
\author{
Lugaoze Feng,
Baoji Wang, 
Guocheng Lv,
Xvnan Li,
Luhua Wang, and
Ye jin
\thanks{Lugaoze Feng, Guocheng Lv and Ye Jin are with the School of Electronics, Peking University, Beijing 100871, China (e-mail: lgzf@stu.pku.edu.cn; lv.guocheng@pku.edu.cn; jinye@pku.edu.cn).}
\thanks{Baoji Wang, Xunan Li, and Luhua Wang are with the National Computer Network Emergency Response Technical Team/Coordination Center of China, Beijing 100029, China (e-mail: bjwang@cert.org.cn; lixunan@cert.org.cn; wlh@cert.org.cn).}}

\maketitle

\begin{abstract}
	The \textit{noisy permutation channel} is a useful abstraction introduced by Makur for point-to-point communication networks and biological storage. While the asymptotic capacity results exist for this model, the characterization of the second-order asymptotics is not available. 
	Therefore, we analyze the converse bounds for the noisy permutation channel in the \textit{finite blocklength} regime. To do this, we present a modified minimax meta-converse for noisy permutation channels by symbol relaxation. To derive the second-order asymptotics of the converse bound, we propose a way to use \textit{divergence covering} in analysis. It enables the observation of the second-order asymptotics and the strong converse via Berry-Esseen type bounds. These two conclusions hold for noisy permutation channels with strictly positive matrices (entry-wise). In addition, we obtain computable bounds for the noisy permutation channel with the binary symmetric channel (BSC), including the original computable converse bound based on the modified minimax meta-converse, the asymptotic expansion derived from our subset covering technique, and the \textit{$\epsilon$-capacity} result. We find that a smaller crossover probability provides a higher upper bound for a fixed finite blocklength, although the $\epsilon$-capacity is agnostic to the BSC parameter. Finally, numerical results show that the normal approximation shows remarkable precision, and our new converse bound is stronger than previous bounds.
\end{abstract}


\section{Introduction}
	\IEEEPARstart{T}{he} \textit{noisy permutation channel}, formally introduced in \cite{makur_coding_2020}, is a point-to-point communication model in which an input $n$-letter undergoes a discrete memoryless channel (DMC) and a uniform random permutation block. It captures the out-of-order arrival of packets in the communication network. Since the uniform permutation block, the only relevant statistic obtained at the receiver is the empirical distribution. This situation often occurs in applications such as communication networks and biological storage systems.  
	
	Generally speaking, the noisy permutation channel is a suitable model for the \textit{multipath routed network.} Consider a point-to-point communication network, where different packets are transmitted through multiple routes, causing different delays. The noisy permutation channel simulates the impact of paths with different delays arriving at the receiver. For example, data packets undergo rerouting in heavily loaded networks and mobile ad hoc networks. The rate delay tradeoffs are analyzed in \cite{walsh_optimal_2008} and \cite{maclaren_walsh_optimal_2009} but do not consider the noise. The insertions, substitutions, deletions, and erasures of symbols are considered in recent work \cite{kovacevic_subset_2013}, \cite{kovacevic_perfect_2015}, \cite{kovacevic_codes_2018}. Their works focus on the perfect codes and minimum distance codes for the permutation channel. Another major application comes from research on \textit{DNA based storage systems} with very high density \cite{yazdi_dna-based_2015}. In \cite{heckel_fundamental_2017}, the authors studied a storage system where data is encoded by DNA molecules. The short DNA molecules are arranged unordered but not corrupted, and the receiver reads the encoded data by shotgun sequencing. Furthermore, a study on the  \textit{noisy shuffling channel} is similar to the noisy permutation channel \cite{shomorony_capacity_2019}. The DNA codewords pass through DMC and are fragmented, and then the fragments are randomly permuted. Such a model simulates that DNA molecules are corrupted by noise at synthesis, sequencing, and during storage.
	
	The main information-theoretic task in the channel coding theorem is to find the maximum communication rate over $n$ uses of a fixed noisy channel $W$. In previous work, several recent advances have been made to the noisy permutation channel with DMC matrices. The original asymptotic bound is given in \cite{makur_coding_2020,makur_bounds_2020}, containing a lower bound that holds in any case and an upper bound that holds for the strictly positive matrices. However, the upper and lower bounds do not match in the case where the rank of the DMC matrix is not equal to the size of the output alphabet. \cite{tang_capacity_2023} improves this upper bound based on the \textit{divergence covering number} and studies the capacity of some non-strictly positive matrices such as $q$-ary erasure channels and Z-channels. Additionally, the capacity area problem of permutation adder multiple-access channels is given in \cite{lu_permutation_2023}, \cite{lu_permutation_2024}.
	
	There are still important issues that need to be resolved. A result known as the strong converse \cite{wolfowitz_notes_1968} in the classical channel coding theorem shows that the asymptotic converse bound holds for every average error probability $\epsilon \in (0,1)$. However, this property has no direct equivalent in the noisy permutation channel. Furthermore, we note that the code lengths of suitable code in such systems are in the order of thousands or hundreds \cite{maclaren_walsh_optimal_2009}, \cite{heckel_fundamental_2017}, invalidating the asymptotic assumptions in classical information theory. The finer asymptotic characterization of the coding theorem gives backoff from channel capacity at a fixed \textit{finite blocklength} \cite{strassen_asymptotic_nodate}, \cite{wang_simple_2009}, \cite{hayashi_information_2009}, \cite{polyanskiy_channel_2010}, \cite{polyanskiy_2010_thesis}, \cite{polyanskiy_saddle_2013}, \cite{tomamichel_tight_2013}. The argument based on binary hypothesis testing allows us to flexibly choose an auxiliary distribution to derive the converse bound for general channels \cite{polyanskiy_channel_2010}, which is termed the minimax meta-converse. For example, the DMC asymptotic upper bound and the third-order asymptotics can be recovered from the meta-converse by Topsoe identity \cite{topsoe1967information} and Berry-Esseen type bounds \cite{tomamichel_tight_2013}, respectively. However, since each message is mapped to an unknown probability distribution, this technique is difficult to apply to noisy permutation channels.

	To fill the mentioned gaps, our main contributions are as follows:
	\begin{enumerate} 
		\item[$\bullet$] First, we develop a modified minimax meta-converse for noisy permutation channels. Using the idea of symbol relaxation \cite{tomamichel_tight_2013} to restrict our computation to the likelihood ratio between the product distribution and the Bayesian distribution, this bound applies to the average error probability. 
		
		\item[$\bullet$] Second, we derive the upper bound for the second-order asymptotics for noisy permutation channels. The key element necessary to complete this proof is our ways based on divergence covering. These ways allow the computation to focus only on the \textit{divergence covering center} of a vector from the transition probability distribution and the divergence covering number of a simplex. Additionally, our asymptotic expansion directly implies the strong converse result for noisy permutation channels.
		
		 \item[$\bullet$] Third, for the noisy permutation channel with the binary symmetric channel (BSC), i.e., the BSC permutation channel, we derive a computable converse bound. We construct a grid-like structure of the set of divergence covering centers using the method in \cite{Jennifer_2022}. This result is obtained using our modified minimax meta-converse on this grid. Although the computational complexity is high, it directly gives the asymptotic behavior of the upper bound. 
		 
		 \item[$\bullet$] Finally, we present the upper bound of the second-order asymptotics and the \textit{$\epsilon$-capacity} for the BSC permutation channel. We estimate the variance by modifying the construction scheme of the grid to cover a subset of the simplex. This result implies a unique property of BSC permutation channels: Although the asymptotic bound is independent of the crossover probability, a smaller BSC crossover probability has a higher upper bound with a fixed finite blocklength $n$. Furthermore, numerical evaluations show that our normal approximation is tight, and our new converse bound is stronger than previous bounds.
	\end{enumerate}
	
	This paper is organized as follows. Section \ref{Section 2} introduces the system model and preliminaries. Section \ref{Section 3} presents our main conclusions: apply minimax converse to the noisy permutation channel and analysis of second-order asymptotics. In section \ref{Section 4}, we use our result to compute the converse bounds of the BSC permutation channel. In Section \ref{Section 5}, we give the numerical evaluation of the BSC permutation channel. We conclude the work and discuss open problems in Section \ref{Section 6}.
	
\subsection{Notation}
	Let $[n] = \{ 1,...,n \}$, $\mathbb{N} = \{1,2,...\}$, and $\mathbb{Z}_+ = \{ 0,1,... \}$. Denote by $1\{ \cdot \}$ the indicator function. Random variables (e.g., $X$) and their realizations (e.g., $x$) are denoted by upper and lower case letters, respectively. Finite sets (e.g., $\mathcal{X}$) are denoted by calligraphic letters. We write $X \sim P_X$ to indicate that the random variable $X$ follows the distribution $P_X$. Let $\mathsf{bin}(n,p)$ denote a binomial distribution with $n$ trials and probability $p$. Denote by $X^n = (X_1,...,X_n)$ and $x^n = (x_1,...,x_n)$ the random vector and its realization in the $n$-th Cartesian product $\mathcal{X}^n$, respectively. Given a finite set $\mathcal{X}$, the cardinality of $\mathcal{X}$ is denoted by $|\mathcal{X}|$. Given a matrix $A$, We use notation $\mathsf{rank}(A)$ to represent the rank of matrix $A$.
	The  probability and expectation are denoted by $\mathbb{P}[\cdot]$, $\mathbb{E}[\cdot]$, where the underlying probability measures will be clear from context. The cumulative distribution function of the standard normal distribution is defined as 
	\begin{equation}
		\Phi(x) = \int_{-\infty}^{x} \frac{1}{\sqrt{2 \pi }}e^{-\frac{t^2}{2}}dt,
	\end{equation}
	and $\Phi^{-1}(\cdot)$ is its inverse function.
	
	A simplex on $\mathbb{R}^K$ is a set of points 
		\begin{equation}
			\Delta_{K-1} \stackrel{\triangle}{=} \left\{ (s_1,s_2,...,s_K)\in \mathbb{R}^k, \ s_k \ge 0, \ \sum_{k=1}^{K} s_k = 1 \right\}.
		\end{equation}
	We emphasize its subsets "n-types" as
			\begin{align}
				\mathcal{P}_n = \bigg\{ P \in \Delta_{K}: \ &P =  \left(\frac{a_1}{n},...,\frac{a_{K}}{n}\right) \nonumber\\
			&	\textbf{{\rm where}} \ a_1,...,a_{K}\in \mathbb{Z}_+ \bigg\}.
			\end{align}	
	For a sequence $x^n = \{ x_1,...,x_n \} \in \mathcal{X}^n$, we denote by $\hat{P}_{x^n} \in \mathcal{P}_n$ the empirical distribution of sequence $x^n$, i.e., $\hat{P}_{x^n}(x) = \frac{1}{n} \sum_{i=1}^{n} 1\{ x_i=x \}$. 		

\section{System Model and Preliminaries} \label{Section 2}
\subsection{Noisy Permutation Channels}
			
	In the noisy permutation channel, the sequence $X^n$ goes through the discrete memoryless channel $W$ to produce $Y^n$. Then, the codeword $Y^n$ passes through a uniformly random permutation part to generate sequence $Y_{\rm Perm}^n$. The sequence $Y^n$ and $Y_{\rm Perm}^n$ take value in $\mathcal{Y}$. We chain all objects together into the following Markov chain:
	\begin{equation}
				X^n \rightarrow Y^n \rightarrow Y_{\rm Perm}^n.
	\end{equation}
	
	According to \cite[Lemma 2]{makur_coding_2020}, sending codeword through a random permutation part and then passing through DMC is equivalent to passing through DMC and then applying random permutation. Therefore, in this paper, we describe the Markov chain of the noisy permutation channel as 
	\begin{equation}
		Z^n \rightarrow X^n \rightarrow Y^n,
	\end{equation}
	where $Z^n$ is the original codeword, and $X^n$ is a random permutation of $Z^n$. For sequences $Z^n$ and $X^n$, we have $Z_i,X_i \in \mathcal{X}$. We use $P_{X^n|Z^n}$ to indicate the random permutation part. To describe this random permutation part, first denote a bijection function as $\lambda: \{ 1,...,n \} \mapsto \{ 1,...,n \}$, drawn randomly and uniformly from the symmetric group $\mathcal{S}_n$ over $\{ 1,...,n \}$. Then we have $X_{\lambda(i)} = Z_i$ for all $i \in \{ 1,...,n \}$. 
	$X^n$ passes through the probability kernel $W$ of DMC and becomes the sequence $Y^n$, where each $Y_i \in \mathcal{Y}$.
	
	 In addition, the DMC can be described as a $|\mathcal{X}| \times |\mathcal{Y}|$ matrix $W$, where $W(y|x)$ denotes the probability that the output $y \in \mathcal{Y}$ occurs given input $x \in \mathcal{X}$.
	We refer to $W $ as strictly positive if all the transition probabilities in $W $ are greater than $0$.
	
\subsection{Preliminaries}
	In the binary hypothesis testing problem, we have two distributions $P$ and $Q$ on a space $\mathcal{X}$. Consider a random variable $V$ which can select one of two distributions $P$ and $Q $ through a random transformation $P_{V|X}: \mathcal{X}\mapsto {0,1}$. Let $V=1$ denote that the test chooses $P$, then we have
	\begin{equation}
		\alpha = P[V=1] = \sum_{x \in \mathcal{X}} P_{V|X}(1|x) P(x),
	\end{equation}
	and the type-II error probability of the test is defined by\footnote{Since the context is well defined, we use probability measures $P[\cdot]$ and $Q[\cdot]$ to represent $\alpha$ and $\beta$ respectively.}	
	\begin{equation}
			\beta = Q[V=1] =\sum_{x \in \mathcal{X}} P_{V|X}(1|x) Q(x).
	\end{equation}
	For a given $\alpha \in (0,1)$, the optimal performance of $\beta$ is 
	\begin{equation} \label{eq:beta_function}
		\beta_\alpha(P ,Q) = \min \sum_{x \in \mathcal{X}} P_{V|X}(1|x)Q(x),
	\end{equation}
	where the minimum is over all probability distributions $P_{V|X}$ satisfying
	\begin{equation}
			P_{V|X}:\sum_{x \in \mathcal{X}}P_{V|X}(1|x) P(x) \ge \alpha.
	\end{equation}	
	The achievability of (\ref{eq:beta_function}) is guaranteed by the Neyman-Pearson lemma.
	 $\beta_{\alpha}(P,Q)$ has the following converse bound \cite{polyanskiy_2023_book}:
	\begin{equation}\label{eq:strong_converse_of_NP_region}
		\alpha - \gamma \beta \le P \left[\log \frac{dP}{dQ} > \log \gamma\right],
	\end{equation}
	where $\gamma > 0$ is arbitrary.
		
	Next, we consider using an $\varepsilon$-net (see \cite{shiryayev_selected_1993} and \cite{yang_information-theoretic_1999}) to cover a probability simplex under Kullback-Leibler (KL) divergence as distance. For distributions $P$ and $Q$ on a discrete alphabet $\mathcal{X}$ , KL divergence is defined as
	\begin{equation}
		D(P\|Q) = \sum_{x \in \mathcal{X}} P(x) \log \frac{P(x)}{Q(x)}.
	\end{equation}
	We can upper bounds KL divergence with $\chi^2$ divergence \cite{csiszar_context_2006}, i.e.,
	\begin{equation} \label{eq:KL_is_upperbounded_by_Chi}
		D(P\|Q) \le \sum_{x \in \mathcal{X}} \frac{(P(x)-Q(x))^2}{Q(x)} = D_{\chi^2}(P\|Q).
	\end{equation}
	
	Let $r$ be the covering radius. The worst-case divergence covering means that for any $P$, there exists a center $Q$ such that $D(P\|Q) \le r$. Let $r>0$, the divergence covering number is defined as
	\begin{align}
		|\mathcal{N}_{r,K}| = \inf \{  & m: \exists\{ Q_1,...,Q_m \} \nonumber\\ 
		&	{\rm s.t.} \max \limits_{P \in \Delta_{K-1}} \min \limits_{Q_i}D(P\|Q_i) \le r \}, 
	\end{align}
	where $\mathcal{N}_{r,K} = \{ Q_1,...,Q_m \}$ is the set of divergence covering centers. According to \cite{Jennifer_2022}, we can construct the set of covering centers based on (\ref{eq:construst_covering_center_3}). Then, there exists a constant $\tau $ such that for $r_0 = r/\tau$ and $\mathcal{N}_{r,K} = \Lambda(r_0)$, the divergence covering number of the simplex $\Delta_{K-1}$ satisfies
	\begin{equation}\label{eq:covering_number_converse}
		|\mathcal{N}_{r,K}|  \le c^{K-1} \left( \frac{K-1}{r} \right)^{\frac{K-1}{2}},
	\end{equation}
	where $c \le 7$. 
	
\subsection{Codes for Noisy Permutation Channels }

		Refers to \cite{polyanskiy_2023_book}, the considered code $\mathcal{C}$ consists of a message set $\mathcal{M}$, an encoder function $f_n: \mathcal{M} \mapsto \mathcal{X}^n$, and a decoder function $g_n: \mathcal{Y}^n \mapsto \mathcal{M} \cup \{ e \}$. We write $\mathcal{X} = [|\mathcal{X}|]$ for the finite input alphabet and $\mathcal{Y} = [|\mathcal{Y}|]$ for the finite output alphabet. Denote by $|\mathcal{M}|$ the cardinality of a code $\mathcal{C}$. The average error probability of a code is defined as
		\begin{align}
				P_e: & = \mathbb{P} [M \not = \hat{M}] \nonumber \\ 
						&= 1- \frac{1}{|\mathcal{M}|} \sum_{m \in \mathcal{M}} P_{X^n|Z^n}(x^n|f_n(m)) W(g_n^{-1}(m)|x^n)
		\end{align}
		where $M $ is the message random variable drawn uniformly from the message set $\mathcal{M}$ and $\hat{M}$ is the estimate of the received sequence. We illustrate this system in Figure \ref{per1}.
			\begin{figure*}[t]
				\normalsize
				\centering
				\includegraphics[width = 0.9\textwidth]{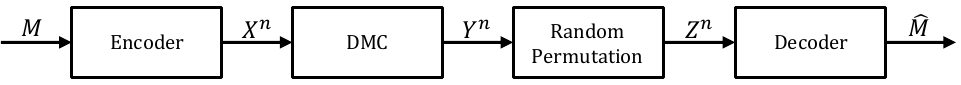}
				\caption{\ Illustration of a communication system with a random permutation followed by a DMC}
				\label{per1}
			\end{figure*}
			
		We define the basic non-asymptotic limit of a channel as $M^*(n,\epsilon) = \max \{M: \exists \ \mathcal{C} \ {\rm such \ that }\ P_e \le \epsilon \}.$
		The rate\footnote{The results in \cite{makur_coding_2020} show that if the capacity is defined in the classical form $\log M^*(n,\epsilon)/n$, the noisy permutation channel will have capacity $0$. In the following section, we will see that $\log n$ can be used to describe the first-order term of noisy permutation channels.} for the encoder-decoder pair $(f_n, g_n)$ is 
			\begin{equation}\label{eq:rate_function}
				R(n,\epsilon) \stackrel{\triangle}{=} \frac{\log M^*(n,\epsilon)}{\log n}.
			\end{equation}
		The $\epsilon$-capacity for noisy permutation channels with $W$ is defined as
		\begin{equation} 
			C_{\epsilon} \stackrel{\triangle}{=} \liminf \limits_{n \rightarrow \infty} \frac{\log M^*(n,\epsilon)}{\log n},
		\end{equation}
		where $\log(\cdot)$ is the binary logarithm (with base 2) throughout this paper. Similar to the classical definition, the strong converse means for every $\epsilon \in (0,1)$, we have
			\begin{equation}
				C_{\epsilon} \le \liminf \limits_{n \rightarrow \infty} \frac{\log M^*(n,\epsilon)}{\log n}.
			\end{equation}
			The weak converse means this bound holds only for $\epsilon \rightarrow 0^+$.
					
		Then, we introduce some useful definitions based on divergence covering. Note that in noisy permutation channels, the codeword is generated by different distributions\footnote{As shown in \cite{makur_coding_2020}, The messages are mapped to different distributions, and the decoder estimates the message using the empirical distribution.}. Denote by $\Theta \subset \Delta_{|\mathcal{X}-1|}$ the set of input distributions. The distribution of $\Theta$ is $P_{\Theta}(\pi)$. 
		
		Let $W_x = W(\cdot|x) \in \Delta_{ \ell}$ be the distribution on $\mathcal{Y}$ if the input is fixed to $x \in \mathcal{X}$, where $\ell = |\mathcal{Y}|-1$ is the dimension of $\Delta_{\ell}$. Similarly, denote by $\pi W: y \rightarrow  \sum_{x \in \mathcal{X}} \pi(x)W(y|x)$ the marginal distribution if $\pi \in \Theta$ is fixed. For a given $x$, consider $Q_{x}^*$ to be the divergence covering center of $W_x$, and we define the conditional divergence $D(W\|Q_{X}^*|\pi) = \sum_{z \in \mathcal{X}} \pi(x) D(W_x\|Q_{x}^*)$. The variance of log-likelihood ratio between $W_x$ and $Q_{x}^*$ is defined as 
		\begin{equation}
			V(W_x\|Q_{x}^*) = \mathbb{E}_{W_x}\left[ \left( \log \frac{W_x}{Q_{x}^*} - D(W_x\|Q_{x}^*) \right)^2 \right].
		\end{equation}
		Moreover, we define the conditional divergence variance is $V(W\|Q_{X}^*|\pi) = \sum_{x \in \mathcal{X}} \pi (x) V(W_x\|Q_{x}^*)$. Let $V(\pi,W) = V(W\|Q_{X}^*|\pi)$, the maximum and minimum variance of the noisy permutation channel based on the divergence covering center is defined as
		\begin{align}
			V_{max} = \max \limits_{\pi \in \Theta} V(\pi, W), \\
			V_{min} = \min \limits_{\pi \in \Theta} V(\pi, W),
		\end{align}
		where $\Theta \subset  \Delta_{|\mathcal{X}-1|}$ is the set of the input probability distributions. The $\epsilon$\textit{-dispersion}\footnote{Here, we set $V_{\epsilon} = V_{max}$ for $\epsilon = 1/2$. Since $\Phi^{-1}(1/2) = 0$, the definition of $V_{1/2}$ is immaterial as far the approximation.} of the noisy permutation channel is
		\begin{equation}
			V_\epsilon = 
				\begin{cases}
					V_{\min} \ & {\mathrm{if}} \  \epsilon < 1/2 \\
					V_{\max} \ & {\mathrm{if}} \  \epsilon \ge 1/2
				\end{cases}.
		\end{equation}
				
		Similarly, we define the third absolute moment of log-likelihood ratio
		\begin{equation}
			T(W_x\|Q_{x}) = \mathbb{E}_{W_x}\left[ \left| \log \frac{W_x}{Q_{x}} - D(W_x\|Q_{x}^*) \right|^3 \right].
		\end{equation}
		
		Consider $x_i \in \mathcal{X}, i \in \{ 1,...,n \}$ and a sequence distribution $Q_{x_i}^* \in \mathcal{N}_r$. Let $Q_{x_i}^*$ be the divergence covering center of $W_{x_i}$. Let $W_{x^n} = \prod_{i=1}^{n} W(\cdot| x_i)$. By the memoryless property of DMC, we have
			\begin{equation}\label{eq:memroyless_property}
				W^n(y^n|x^n) = \prod_{i=1}^{n} W(y_i|x_i).
			\end{equation}

\section{Converse Bounds for Noisy Permutation Channels} \label{Section 3}
	
	In this section, we develop a upper bound for second-order asymptotics for the noisy permutation channel with strictly positive matrices. The proof consists of three parts, each detailed in one of the following sections. In Subsection \ref{Subsection 3_meta}, we derive a modified minimax meta-converse for noisy permutation channels that involves a Bayes mixture density based on divergence covering. In Subsection \ref{Subsection 3_asy}, we prove some useful properties that we need later. Finally, we give our main result in Subsection \ref{Subsection 3_main}.

\subsection{A Modified Minimax Meta-Converse}\label{Subsection 3_meta}
	In this subsection, we present a modified minimax meta-converse \cite[Theorem 27]{polyanskiy_channel_2010} with an arbitrary set of input distributions, and this bound applies to the average error probability. 
	
	First, we introduce a Bayes mixture density. Let $\Delta_{ \ell}^*$ be the convex hull of $\{ \pi W: \pi \in \Theta \}$. For convenience of expression, denote by $\mathcal{N}_r$ the set of divergence covering centers of simplex $\Delta_\ell^*$. Consider constructing the subset of $\mathcal{N}_r$ in the following way:
					\begin{enumerate}
					\item Pick $Q_{x}$ such that $\mathcal{N}_{1} = \{ Q_{x} : Q_{x} \in \mathcal{N}_r, x \in [|\mathcal{X}|]\}$.
					\item Pick $Q_{x}$ such that $\mathcal{N}_{2} = \{ Q_{x} : Q_{x} \in  \mathcal{N}_r\setminus \mathcal{N}_{1}, x \in [|\mathcal{X}|] \}$.
					
					\
					$\cdots$
					\item Pick $Q_{x}$ such that $\mathcal{N}_{E-1} = \{ Q_{x}: Q_{x} \in \mathcal{N}_r \setminus \cup_{e=1}^{E-2}\mathcal{N}_{e}, x \in [|\mathcal{X}|] \}$.
					\item Pick elements from $\cup_{k=1}^{E-1}\mathcal{N}_{k}$ and add them to $\mathcal{N}_r \setminus \cup_{e=1}^{E-1}\mathcal{N}_{e}$, such that the resulting set satisfies $|\mathcal{N}_r^*| = |\mathcal{X}|$. Let $\mathcal{N}_{E} = \{ Q_{x}: Q_{x} \in \mathcal{N}_r^*, x \in [|\mathcal{X}|]\} $.
					\end{enumerate} 
					We can at most get $E = \big \lceil |\mathcal{N}_r| / |\mathcal{X}| \big \rceil$ such subsets. Let $\mathcal{O} = \cup_{e=1}^E \mathcal{N}_e$, we get a Bayes mixture density
					\begin{align}
						\hat{Q}_{Y^n} (y^n) & =   \frac{1}{E} \sum_{\mathcal{N}_e \in \mathcal{O}} \prod_{Q_x \in \mathcal{N}_e} \prod_{i=1}^{n\hat{P}_{x^n}(x)} Q_{x}(y_i),\label{eq:bayes_distrubution}
					\end{align}	
	
	Then, we give our new converse bound as follows:
	\begin{proposition} \label{proposition:meta_converse_of_permutation channel}
			For the noisy permutation channel with strictly positive matrices,  let $W$ is any channel from $\mathcal{X}$ to $\mathcal{Y}$ and $\mathcal{N}_r$ is any set of divergence covering centers, we have
			\begin{equation} \label{eq:meta_converse}
				\log M^*(n,\epsilon) \le \sup \limits_{\pi \in \Theta} \sup \limits_{x^n \in \mathcal{X}^n} - \log \beta_\alpha(W_{x^n},Q_{Y^n}), 
			\end{equation}
			where we set $Q_{Y^n}=\hat{Q}_{Y^n}$.
		\end{proposition}
		Proposition \ref{proposition:meta_converse_of_permutation channel} is proved in Appendix \ref{Appendix A}. We deal with distributions that are $n$-fold products of a fixed distribution since the permutation block acts on the entire input sequence. Additionally, we refer to the idea of symbol relaxation  \cite{tomamichel_tight_2013} to make our result apply to the average error formulation. We can apply the result of the divergence covering number to Proposition \ref{proposition:meta_converse_of_permutation channel} since it holds for any set of divergence covering centers.
	
\subsection{Asymptotics and Bounds Based on Divergence Covering}	\label{Subsection 3_asy}
	To derive the asymptotic expansion of Proposition \ref{proposition:meta_converse_of_permutation channel}, we will be concerned with the asymptotic behavior of the $\beta$ function and bounds of $V_n$ and $T_n$. 
	
	The following Lemma will be particularly useful.
	\begin{lemma}\label{lemma:prob_change}
				In the notation of Proposition \ref{proposition:meta_converse_of_permutation channel}, we have
				\begin{align}
					&\mathbb{P}  \left[  \log \frac{W_{x^n}}{Q_{Y^n}} \ge \log \gamma \right] \nonumber \\
					& \ \ \ \le \mathbb{P} \left[ \sum_{i=1}^{n} \log \frac{W_{x_i}}{Q_{x_i}^*} \ge \log \gamma -  \log E\right], 
				\end{align}
				where $\mathbb{P}$ means under any probability $P^n$.
	\end{lemma}	
	Lemma \ref{lemma:prob_change} is proved in Appendix \ref{Appendix B}. It allows us to only deal with the covering center of $W_x$ and the divergence covering number of the $\ell$-dimensional simplex. 
	
	Then, we define
			\begin{align}
					D_n = & \frac{1}{n} \sum_{i=1}^{n} D(W_{x_i} \| Q_{x_i}^*), \
					V_n  =  \frac{1}{n} \sum_{i = 1}^{n} V(W_{x_i} \| Q_{x_i}^*), \\
					T_n = & \frac{1}{n} \sum_{i = 1}^{n} T(W_{x_i} \| Q_{x_i}^*).
			\end{align}
	Using Lemma \ref{lemma:prob_change}, we have the following asymptotic expansions.
	\begin{lemma} \label{lemma:asymptotic_bound_of_beta}
		In the notation of Proposition \ref{proposition:meta_converse_of_permutation channel}, if $V_n >0$, for any $\Delta > 0$, we have
		\begin{align} \label{eq:Berry_Esseen_bound_of_beta}
			&\log \beta_{\alpha}(W_{x^n} , Q_{Y^n}) \nonumber\\
							&\ \ \ \ge -\log E - nD_n + \log \frac{\Delta}{\sqrt{n}} \nonumber\\
							& \ \ \ \ \ \ + \sqrt{nV_n} \Phi^{-1}\left( \alpha - \frac{6T_n}{\sqrt{nV_n^3}} - \frac{\Delta}{\sqrt{n}}\right). 
		\end{align}
		In any case, we have
		\begin{align} \label{eq:Chebyshev_bound_of_beta}
			\log \beta_{\alpha}(W_{x^n} , Q_{Y^n}) & \ge -\log E - nD_n \nonumber\\ 
			 & \ \ \ - \sqrt{\frac{2nV_n}{1 - \epsilon}} + \log \frac{\alpha}{2}.
		\end{align}
	\end{lemma}
	Lemma \ref{lemma:asymptotic_bound_of_beta} is provided in Appendix \ref{Appendix C}. It can be construed as a variant of \cite[Lemma 58]{polyanskiy_channel_2010}, except we focus on the asymptotic behavior between $W_x$ and its divergence covering center.
	
	Then, we present the following lemma to illustrate some properties of $V_n$ and $T_n$.
	\begin{lemma} \label{lemma:properities_Vn_Tn}
		For any $y \in \mathcal{Y}$, we have
			\begin{equation}
				V_n \le \frac{C_1}{n}, \ T_n \le \frac{C_2}{n^{3/2}},
			\end{equation}
			where $C_1$ and $C_2$ are constant.
	\end{lemma}
	Lemma \ref{lemma:properities_Vn_Tn} is proved in Appendix \ref{Appendix D}. It illustrates that the variance of the log-likelihood ratio between $W_x$ and $Q_x$ vanishes at a rate of $O(1/n)$, and the third absolute moment vanishes at a rate of $O(1/n^{3/2})$.

\subsection{Main Results}	\label{Subsection 3_main}
	
	Now, through the above auxiliary conclusions, we can derive the main result of this section, which is the upper bound for the second-order asymptotics for the noisy permutation channel in the finite blocklength regime. Please refer to Appendix \ref{Appendix E} for the proof.
		\begin{theorem}\label{thm:asymptotic_bound_of_permutation_channel}
				For the noisy permutation channel with strictly positive matrices, the average error probability of code $\mathcal{C}$ satisfies:
				\begin{align}
					\log M^*(n,\epsilon) \le & \ \frac{ \ell}{2} \log n + \sqrt{nV_{\epsilon}} \Phi^{-1}(\epsilon) \nonumber \\
					&+ O(\log \log n).
				\end{align}
		\end{theorem}
		\begin{remark}\label{remark:main_result}
			In the proof of Theorem \ref{thm:asymptotic_bound_of_permutation_channel}, we find that for any $G_1>0$, there exists an $N_0$ sufficiently large, such that $O(\log \log n) = G_1 \log \log n + O(1)$ holds for $n \ge N_0$. Since the constant $G_1$ cannot be determined, we use $O(\log \log n)$ to represent the third-order and $O(1)$ terms. 
		\end{remark}
		
		Using Theorem \ref{thm:asymptotic_bound_of_permutation_channel}, we can obtain the strong converse directly.
		\begin{corollary} \label{corollary:strong_converse_of_permutation_channel}
			For the noisy permutation channel with strictly positive matrices, the $\epsilon$-capacity must satisfy
			\begin{equation}
				C_\epsilon \le \frac{\ell}{2}.
			\end{equation}
		\end{corollary}
		\begin{proof}
			Using Lemma \ref{lemma:properities_Vn_Tn}, we can find a constant $F_0$ such that $V_{min} \le F_0/n$ and a constant $F_1$ such that $V_{max} \le F_1/n$. Then we have $\sqrt{nV_{min}} \le F_0$ and $\sqrt{nV_{max}} \le F_1$. Substitute these into Theorem \ref{thm:asymptotic_bound_of_permutation_channel}, we get the following result, which implies the strong converse:
						\begin{equation}
							\log  M^*(n, \epsilon) \le \frac{|\mathcal{Y}-1|}{2} \log n + o(\log n).
						\end{equation}
		\end{proof}
		
		\begin{remark}
			Corollary \ref{corollary:strong_converse_of_permutation_channel} is stronger than the converse bound in \cite{makur_coding_2020}. If $\mathsf{rank}(W) = |\mathcal{Y}|$, it is also stronger than the converse bound in \cite{tang_capacity_2023}. As we will see in Corollary \ref{corollary:epsilon_capacity_of_BSC_permutation_channel} in Section \ref{Section 4},  extending this bound to $\epsilon$-capacity is straightforward with restriction $\mathsf{rank}(W) = |\mathcal{Y}|$.
		\end{remark}

\section{Explicit Bounds  for BSC Permutation Channels} \label{Section 4}
	To complement our main result in Section \ref{Section 3}, we compute the explicit bound of the BSC permutation channel using the results in Section \ref{Section 3}. We refer to the scheme in \cite{Jennifer_2022} to construct a grid-like structure and analyze the bounds based on this structure. 
	
	In $2$-dimensional case, we set $\mathcal{N}_r = \bigcup_{\mu =1}^{|\mathcal{N}_{r}|} (q_\mu,1-q_\mu)$. The following result can be obtained by Proposition \ref{proposition:meta_converse_of_permutation channel}. We refer to the idea in \cite[Theorem 35]{polyanskiy_channel_2010} to complete the proof. It can be considered as a direct application of Neyman-Pearson Lemma.
	\begin{theorem}\label{thm:meta_converse_of_BSC_permutation_channel}
		For the BSC permutation channel with crossover probability $\delta$, we have
		\begin{equation}
			\log M^*(n,\epsilon) \le -\log \beta_{\alpha}^n.
		\end{equation}
		 $ \beta_{\alpha}^n$ is defined as
		\begin{align}
			\beta_{\alpha}^n =& \lambda \binom{n}{T}Q_{n,T} + \sum_{t=0}^{T-1}\binom{n}{t}Q_{n,t}, \\
			Q_{n,t} &= \frac{1}{|\mathcal{N}_r|-2K}\sum_{\mu = K}^{|\mathcal{N}_r|-K} q_\mu^t(1-q_\mu)^{n-t},
		\end{align}
		where $\lambda \in (0,1)$, integer $K$ such that $q_{K} < \delta$ and the integer $T$ is defined by
		\begin{align}
			\alpha =  \lambda \binom{n}{T}\delta^T(1-\delta)^{n-T}+ \sum_{t=0}^{T-1}\binom{n}{t}\delta^T(1-\delta)^{n-T}.
		\end{align}
	\end{theorem}
	\begin{Proof}
		Without loss of generality, let $\delta < 1/2$. We construct the set of covering centers based on (\ref{eq:construst_covering_center_2}), where we set $r_0 = \frac{1}{\tau n}$. There exists a $\tau$ such that the covering radius is less than $1/n$. For any $\mathcal{N}_e \in \mathcal{O}$, let $Q_{0}=(q_\mu,1-q_\mu), Q_{1}=(1-q_\mu,q_\mu)$. For any $\pi \in \Theta$, the output distribution is within $ (\delta,1-\delta)$. We only need to take $q_\mu \ge \delta$, i.e., 
		\begin{equation}
			Q_{Y^n}(|y^n|=t) = \frac{1}{|\mathcal{N}_r|-2K}\sum_{\mu = K}^{|\mathcal{N}_r|-K} q_\mu^t(1-q_\mu)^{n-t},
		\end{equation}
		where $K$ such that $q_K < \delta$, $|y^n|$ is the Hamming weight of the binary sequence $y^n$. For convenience of expression, let $Q_{n,t} = Q_{Y^n}$.
		
		The transition probabilities of BSC is
				\begin{equation}
					W^n (y^n|x^n)= \delta^{y^n-x^n} (1-\delta)^{n-|y^n-x^n|}.
				\end{equation}
		Now, fix crossover probability, we find the Hamming weight $|Y^n-X^n|$ is a sufficient statistic between $W^n $ and $Q_{Y^n}$. The optimal performance of type II error probability can be obtained as
		\begin{equation}
			\beta_{\alpha}^n = \lambda \binom{n}{T}Q_{n,T} + \sum_{t=0}^{T-1}\binom{n}{t}Q_{n,t}, 
		\end{equation}
		where the $\lambda $ is uniquely determined by the condition
		\begin{equation}
			\alpha =  \lambda \binom{n}{T}\delta^T(1-\delta)^{n-T}+ \sum_{t=0}^{T-1}\binom{n}{t}\delta^T(1-\delta)^{n-T}.
		\end{equation}
		Then, by Proposition \ref{proposition:meta_converse_of_permutation channel}, we have
		\begin{equation}
			\log M^*(n,\epsilon) \le -\log \beta_{\alpha}^n.
		\end{equation}
		{\hfill $\square$}
	\end{Proof}
	
	Next, we use Theorem \ref{thm:asymptotic_bound_of_permutation_channel} to compute the second-order asymptotics of the BSC permutation channel. Since the variance is defined based on the set of divergence covering centers, obtaining the explicit bound requires the grid constructed by (\ref{eq:construst_covering_center_2}). We get the second result by using our subset covering technique to construct the set of divergence covering centers for the subset of the simplex.
	\begin{theorem} \label{thm:asymptotic_bound_of_BSC_permutation_channel}
			For the BSC permutation channel, if $\epsilon \in (0,1/2)$, the average error probability of a code $\mathcal{C}$ satisfies:
			\begin{align} \label{eq:asymptotic_bound_of_BSC_permutation_channel_1}
					\log M^*(n,\epsilon) \le & \ \frac{1}{2} \log n + \sqrt{nV_{min}(\delta)} \Phi^{-1}(\epsilon) \nonumber \\ 
							&+ O(\log \log n),
			\end{align}
			and if $\epsilon \in \left[1/2,1\right)$, 
			\begin{align} \label{eq:asymptotic_bound_of_BSC_permutation_channel_2}
				\log M^*(n,\epsilon) \le & \ \frac{1}{2} \log n + \sqrt{nV_{max}(\delta)} \Phi^{-1}(\epsilon) \nonumber \\ 
				&+ O(\log \log n).
			\end{align}
			where $\tau = 2/\delta$, the definition of $V_{min}(\delta)$ and $V_{max}(\delta)$ see (\ref{eq:lower_bound_of_variance}) and (\ref{eq:upper_bound_of_variance}). 
		\end{theorem}
		Theorem \ref{thm:asymptotic_bound_of_BSC_permutation_channel} is proved in Appendix \ref{Appendix F}. In addition, we have the following corollary to illustrate the $\epsilon$-capacity of BSC permutation channel.
		\begin{corollary} \label{corollary:epsilon_capacity_of_BSC_permutation_channel}
			For the BSC permutation channel, the $\epsilon$-capacity satisfies
			\begin{equation} \label{epsilon capacity of BSC permutation channel1}
				C_\epsilon = \frac{1}{2}.
			\end{equation}
		\end{corollary}
		\begin{Proof}
			The achievability part has been proved in \cite{makur_coding_2020}. Using Corollary \ref{corollary:strong_converse_of_permutation_channel}, we get $C_\epsilon \le 1/2$. We combine the achievability and converse parts to obtain (\ref{epsilon capacity of BSC permutation channel1}).
			{\hfill $\square$}
		\end{Proof}
		
		\begin{remark} \label{remark:upper_property}
			Theorem \ref{thm:asymptotic_bound_of_BSC_permutation_channel} and Corollary \ref{corollary:epsilon_capacity_of_BSC_permutation_channel} imply that, although the asymptotic capacity of the BSC permutation channel is unrelated to the crossover probability, a smaller crossover probability has a higher upper bound at a fixed finite blocklength $n$.
		\end{remark}

\section{Numerical Results}\label{Section 5}
In this section, we perform numerical evaluations of the BSC permutation channel. The $\epsilon$-capacity of the BSC permutation channel is illustrated in Corollary \ref{corollary:epsilon_capacity_of_BSC_permutation_channel}, i.e., $C_\epsilon = 1/2$. 
	
	First, according to Theorem \ref{thm:meta_converse_of_BSC_permutation_channel}, we have the following upper bound:
	\begin{equation} \label{eq:numerical_converse}
		R(n,\epsilon) \le \frac{- \log \beta_{\alpha}^n}{\log n}.
	\end{equation}
	Then, we omit the third-order term. According to Theorem \ref{thm:asymptotic_bound_of_BSC_permutation_channel}, we obtain the following normal approximation:
			\begin{align} \label{eq:numerical_approximation}
					R(n,\epsilon) \lessapprox \frac{1}{2} + \frac{\sqrt{nV(\delta)} \Phi^{-1}(\epsilon)}{\log n} ,
			\end{align}
			where 
			\begin{equation}
			V(\delta) = 
				\begin{cases}
					V_{\min} \ & {\mathrm{if}} \  \epsilon < 1/2 \\
					V_{\max} \ & {\mathrm{if}} \  \epsilon > 1/2
				\end{cases}.
			\end{equation}
			The definition of $V_{\min}(\delta)$ and $V_{\max}(\delta)$ see (\ref{eq:lower_bound_of_variance}) and (\ref{eq:upper_bound_of_variance}).
				
	As we mentioned in Remark \ref{remark:main_result}, for any $G_1>0$, we have $O(\log \log n) = G_1 \log \log n +O(1)$ for $n$ sufficiently large. Let $G_1 = \frac{1}{16}$, we have 
	\begin{equation} \label{eq:numerical_third}
		R(n,\epsilon) \lessapprox  \frac{1}{2} + \frac{\sqrt{nV(\delta)} \Phi^{-1}(\epsilon)}{\log n} + \frac{\log \log n.}{16 \log n}.
	\end{equation}
\subsection{Precision of the Normal Approximation}
	We compare the bounds (\ref{eq:numerical_converse}), (\ref{eq:numerical_approximation}), and (\ref{eq:numerical_third}) with different crossover probability $\delta$ and average error probability $\epsilon$ in  Fig. \ref{fig:converse_comparison_1}, Fig. \ref{fig:converse_comparison_2}, Fig. \ref{fig:converse_comparison_3}, and Fig. \ref{fig:converse_comparison_4}. We note that (\ref{eq:numerical_approximation}) and (\ref{eq:numerical_third}) indicate the asymptotic behavior of (\ref{eq:numerical_converse}). Furthermore, (\ref{eq:numerical_third}) shows remarkable precision by taking a suitable constant $G_1$. Therefore, we can replace the complex computations of Theorem \ref{thm:meta_converse_of_BSC_permutation_channel} by Theorem \ref{thm:asymptotic_bound_of_BSC_permutation_channel}. 
		
	\begin{figure}[t]
				\centering
				\includegraphics[width = 0.45\textwidth]{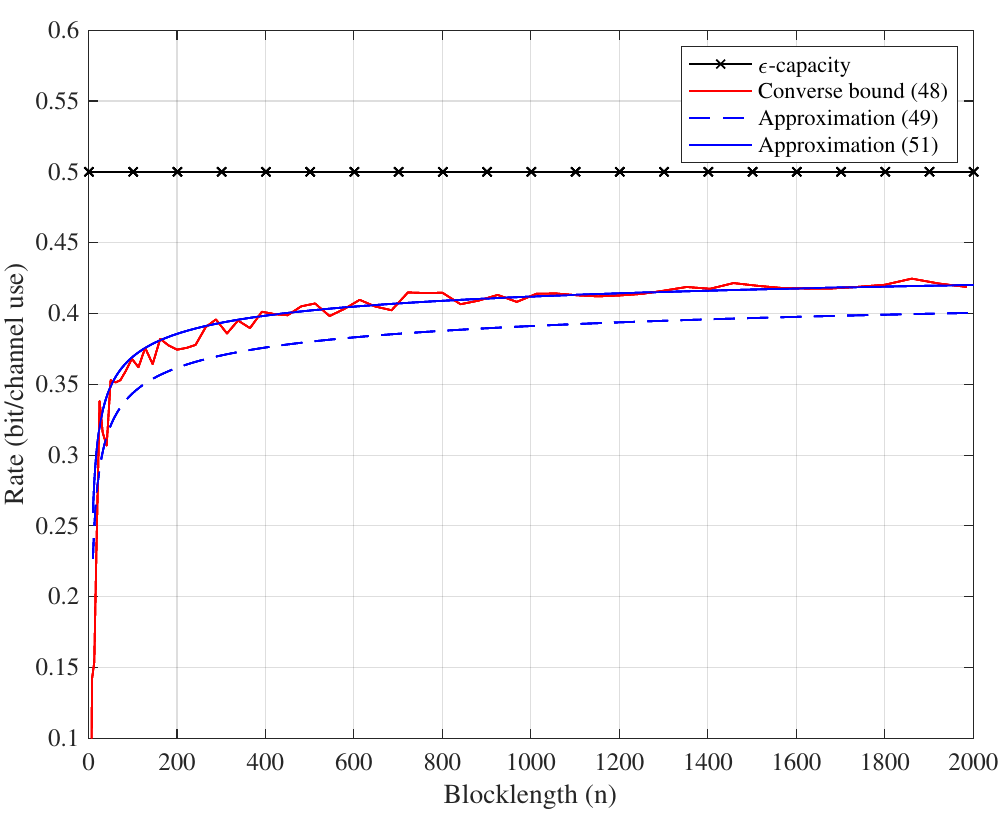}
				\captionsetup{font=footnotesize}
				\caption{\ Comparison of the converse bound and its approximation for the BSC permutation channel with crossover probability $\delta = 0.11$ and average error probability $\epsilon = 10^{-3}$.}
				\label{fig:converse_comparison_1}
		\end{figure}
			\begin{figure}[t]
				\centering
				\includegraphics[width = 0.45\textwidth]{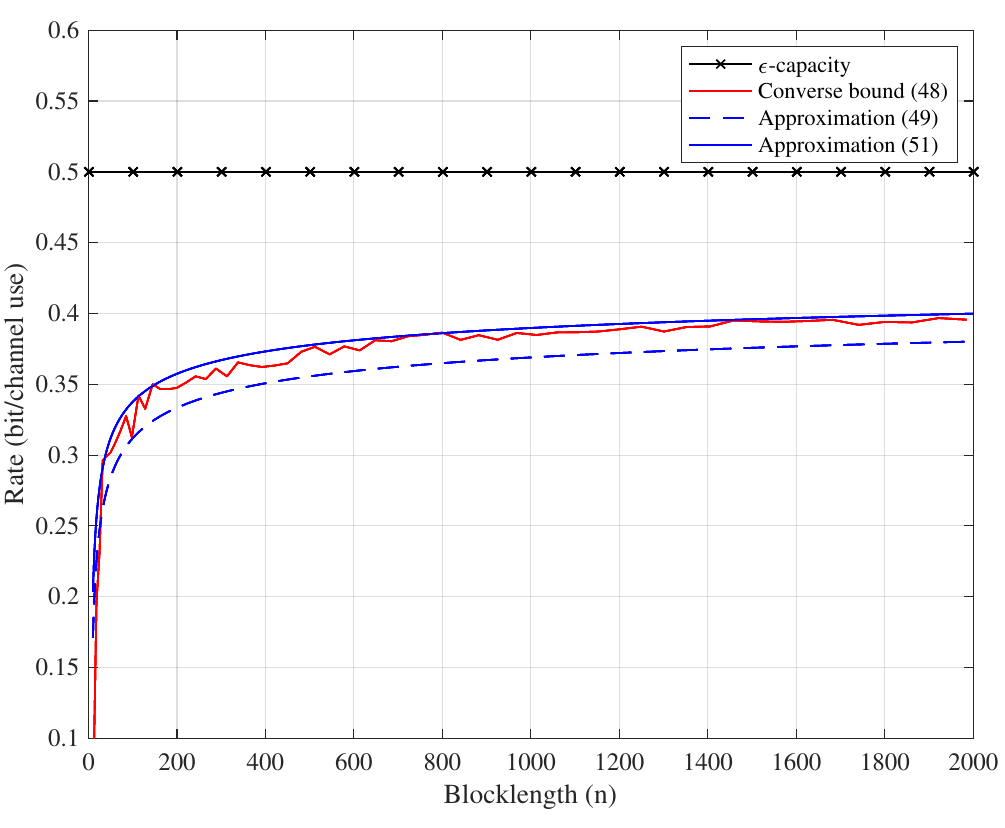}
				\captionsetup{font=footnotesize}
				\caption{\ Comparison of the converse bound and its approximation for the BSC permutation channel with crossover probability $\delta = 0.11$ and average error probability $\epsilon = 10^{-4}$.}
				\label{fig:converse_comparison_2}
			\end{figure}					   
		
	\begin{figure}[t]
					\centering
					\includegraphics[width = 0.45\textwidth]{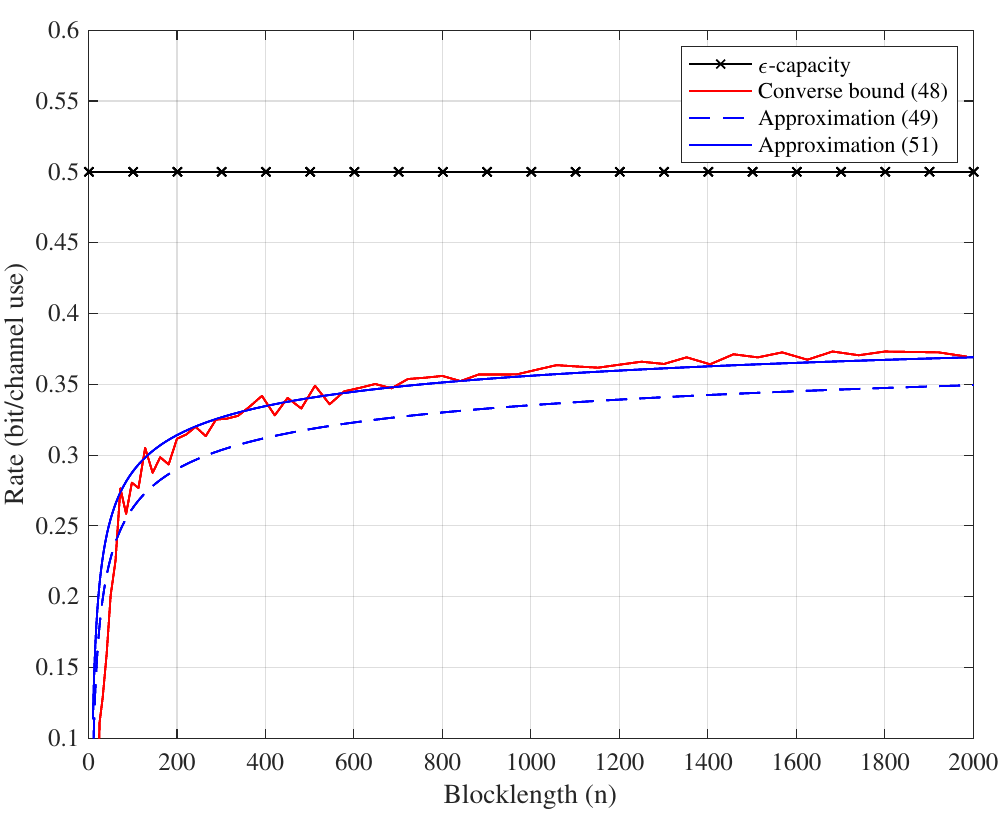}
					\captionsetup{font=footnotesize}
					\caption{\ Comparison of the converse bound and its approximation for the BSC permutation channel with crossover probability $\delta = 0.22$ and average error probability $\epsilon = 10^{-3}$.}
					\label{fig:converse_comparison_3}
		\end{figure}				
		
		\begin{figure}[t]
			\centering
			\includegraphics[width = 0.45\textwidth]{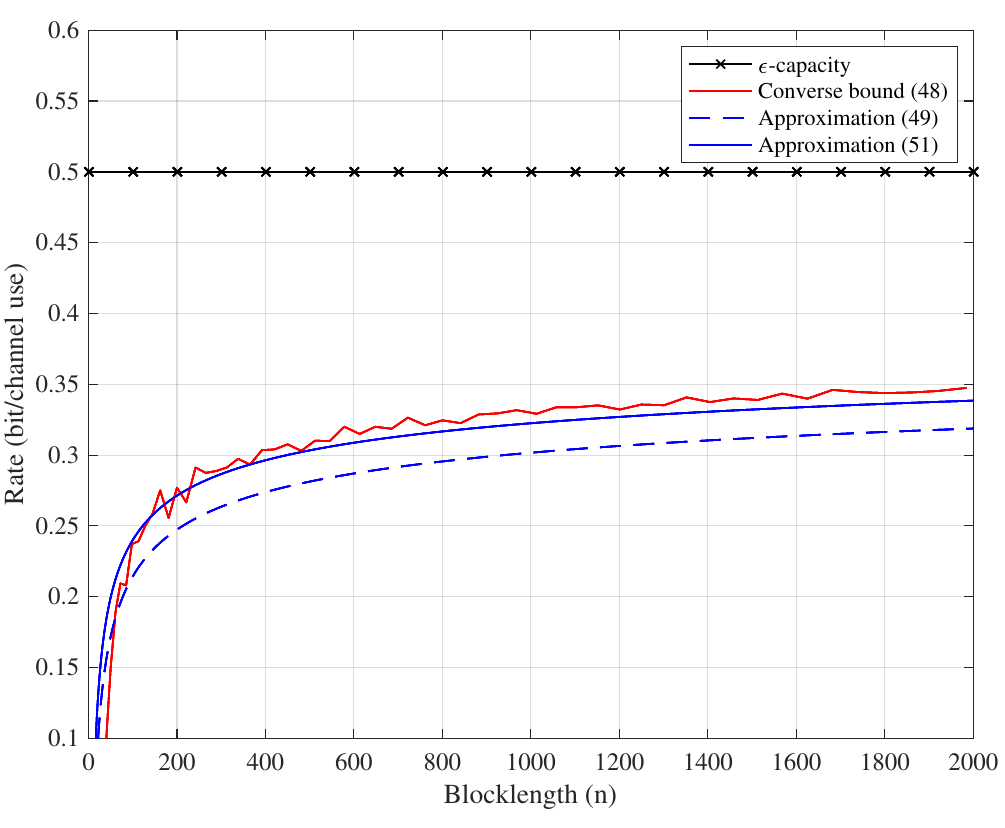}
			\captionsetup{font=footnotesize}
			\caption{\ Comparison of the converse bound and its approximation for the BSC permutation channel with crossover probability $\delta = 0.22$ and average error probability $\epsilon = 10^{-4}$.}
			\label{fig:converse_comparison_4}
		\end{figure}		
		
		\begin{remark}
			The computation results of  Theorem \ref{thm:meta_converse_of_BSC_permutation_channel} show fluctuations because the explicitly constructed grid-like structure is suboptimal. Studying the optimal form of the set of divergence covering centers can improve this result. However, this result still shows the asymptotic behavior of the BSC permutation channel very well. 
		\end{remark}		
		
\subsection{Properties of the Normal Approximation}
		Next, In Fig. \ref{fig:approximation_comparison}, let $\epsilon = 10^{-3}$, we compared the normal approximation (\ref{eq:numerical_approximation}) of the BSC permutation channel with different crossover probabilities. We make two conclusions from the plot:
			\begin{enumerate}
			\item As we discussed in Remark \ref{remark:upper_property}, when the blocklength $n$ is fixed and finite, a smaller crossover probability has a higher upper bound.
			\item The change of slope of the (\ref{eq:numerical_approximation}) is similar to the classical BSC. However, it should be noted that the rate of noisy permutation channel is defined as $\log M^*(n,\epsilon) / \log n$, which is different from the classical definition. 
			\end{enumerate}
			 
			\begin{figure}[t]
				\centering
				\includegraphics[width = 0.45\textwidth]{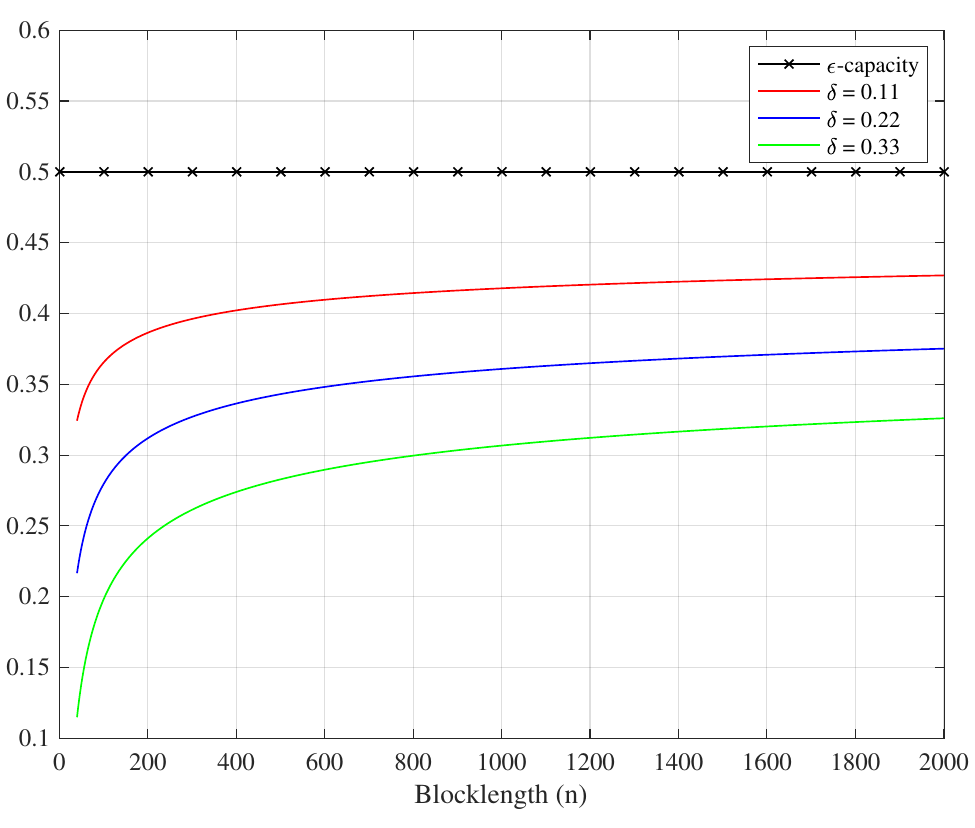}
				\captionsetup{font=footnotesize}
				\caption{\ Comparison of normal approximation of converse bound for the BSC permutation channel with average error probability $\epsilon = 10^{-3}$: different crossover probability. }
				\label{fig:approximation_comparison}
			\end{figure}

\subsection{Comparison with Previous Converse Bounds}
	In this subsection, we compare our new bound with the previous bound. We consider the BSC permutation channel, according to \cite[Theorem 2]{makur_coding_2020}, we have the following upper bound hold for every $n$:
	\begin{equation}\label{eq:num_pre}
		R(n,\epsilon) \le \frac{1 + I(X^n;\hat{P}_{Y^n})}{(1 - \epsilon)\log n},
	\end{equation}
	where
	\begin{align}
		I(X^n;\hat{P}_{Y^n})& \le \log(n+1) - \frac{2c(\delta)}{n-2} \nonumber\\
		& \ \ \ - \frac{1}{2} \log \left( 2 \pi e \delta(1-\delta) \frac{n-2}{2} \right).
	\end{align}
	To determine $c(p)$, let $\hat{X} \sim \mathsf{bin}(n,\delta)$, then we have 
	\begin{equation}
		c(\delta) = n \cdot\left| H(\hat{X}) - \frac{1}{2} \log (2 \pi e n \delta (1-\delta)) \right|.
	\end{equation}
	We consider the following setting: $\delta = 0.11$ and $\epsilon = 10^{-3}$. The bound (\ref{eq:numerical_converse}) and bound (\ref{eq:num_pre}) are compared in Fig. \ref{fig:comparison_with_pre}. We see that bound (\ref{eq:num_pre}) converges to the $\epsilon$-capacity above $1/2$, and our new bound gives better results than (\ref{eq:num_pre}) for every $n$.
	
	\begin{figure}[t]
					\centering
					\includegraphics[width = 0.45\textwidth]{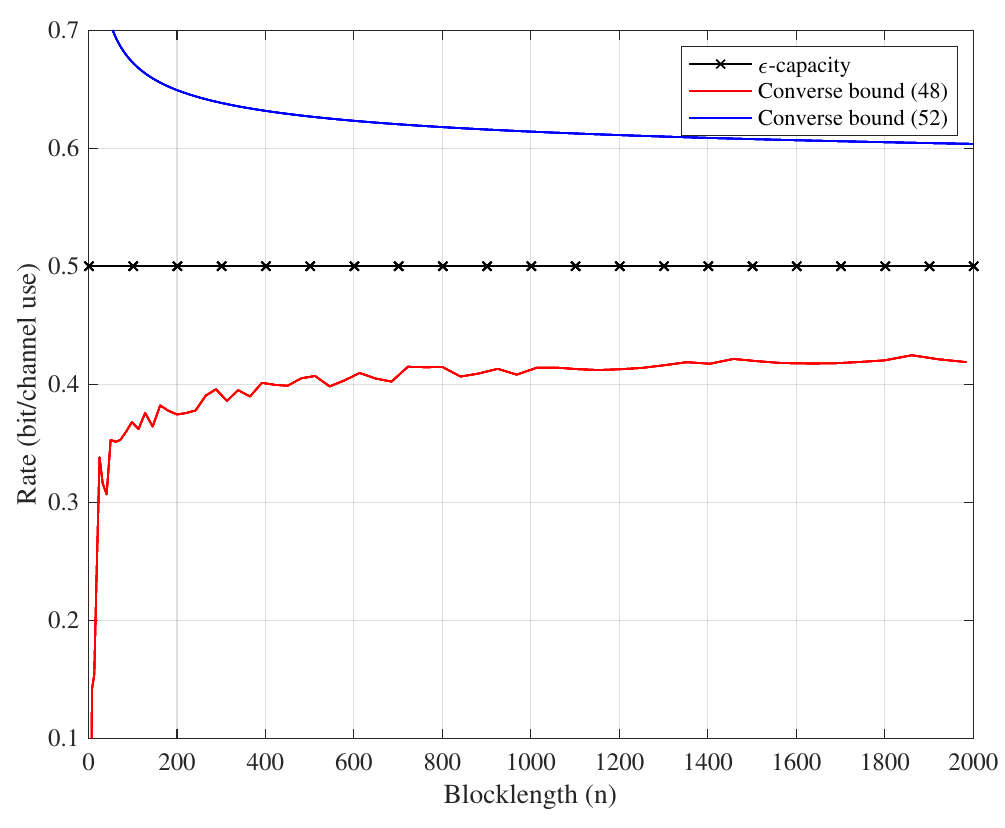}
					\captionsetup{font=footnotesize}
					\caption{\ Comparison of new and previous converse bound for the BSC permutation channel with crossover probability $\delta = 0.11$ and average error probability $\epsilon = 10^{-3}$. }
					\label{fig:comparison_with_pre}
	\end{figure}

\section{Conclusion}\label{Section 6}
		
		In summary, we generalize the minimax meta-converse to derive the converse bound for the noisy permutation channel. The key to the proof is that we adopt the ideas of symbol relaxation. Based on the divergence covering, we observe the second-order asymptotics and the strong converse. Additionally, we also determine the explicit converse bound, normal approximation of converse bound, and the $\epsilon$-capacity of the BSC permutation channel. We find a unique property of the BSC permutation channel: at fixed finite blocklength $n$, a smaller BSC crossover probability has a higher upper bound. Numerical evaluations indicate that the approximation can replace the complex computations of the converse bound, and our new bound is stronger than converse bound in \cite{makur_coding_2020}.
		
		Finally, we propose some future research directions. Firstly, \cite{Jennifer_2022} shows that the upper bound can be strengthened to a tighter form so that even in the case of $\mathsf{rank}(W) \neq |\mathcal{Y}|$, the upper bound matches the lower bound. Finding the upper bound at finite blocklengths satisfies this form is desirable. Secondly, a promising direction is to develop the achievability bound and its approximation for the noisy permutation channel. Lastly, future work may adapt our result to the noisy permutation channel with general DMC matrices, where the transition probability in $W$ can be equal to $0$ instead of greater than $0$.

\appendices
	
\section{Proof of Proposition \ref{proposition:meta_converse_of_permutation channel}}	\label{Appendix A}
		To prove the Proposition \ref{proposition:meta_converse_of_permutation channel}, we first need the following Lemma\footnote{This property was proposed by Makur \cite{makur_coding_2020}. }.
	\begin{lemma}\label{lemma:permutation_do_not_change_distribution}
			Fixed $\pi \in \Theta$, generate codeword $Z^n \stackrel{i.i.d.}{\sim} \pi$. After $Z^n$ undergoes random independent permutation, we have $X^n \stackrel {i.i.d.}{\sim} \pi$.
		\end{lemma}

	{ \noindent\textit{Proof of Proposition \ref{proposition:meta_converse_of_permutation channel}.} 	We assume that $W$ is strictly positive, i.e., for any $x \in \mathcal{X}$, we have $W_x>0$. In the noisy permutation channel, since the random permutation, the order of symbols of codewords conveys no information, and the type $\mathcal{P}_n$ is the only statistical information. Therefore, we need to analyze codewords generated by different probability distributions. Consider a Markov chain $M \rightarrow \Theta \rightarrow Z^n \rightarrow X^n \rightarrow Y^n \rightarrow \hat{M}$, where $\Theta \subset \Delta_{|\mathcal{X}-1|}$ is the set of the input probability distribution. Randomly and independently generate codeword $Z^n$ according to $\pi$, i.e., $ Z^n \stackrel{i.i.d.}{\sim} \pi$. Then, sent codewords to the noisy permutation channel. We use the joint probability density $P_{M,\Theta ,Z^n,X^n,Y^n, \hat{M}}$ to describe this Markov chain. }

	Define random variable $V = 1 \{ M = \hat{M} \}$ and random test $P_{V|\Theta ,Z^n,X^n,Y^n}$, according to meta-converse in \cite[Theorem 26 ]{polyanskiy_channel_2010}, we have
	\begin{align}
		&\mathbb{E}_{P_{\Theta,Z^n,X^n} \times W^n } [V=1] \ge 1 - \epsilon, \\
		&\mathbb{E}_{P_{\Theta,Z^n,X^n} \times  Q_{Y^n}} [V=1] = \frac{1}{|\mathcal{M}|}.
	\end{align}
	Through the definition of $\beta$ function, we readily see that
	\begin{equation}\label{eq:message_and_beta}
		\frac{1}{|\mathcal{M}|} \ge \beta_{\alpha}(P_{\Theta, Z^n,X^n} W^n ,P_{\Theta,Z^n,X^n} Q_{Y^n}).
	\end{equation}
	Notice that
	\begin{align}
		&\beta_{\alpha}(P_{\Theta, Z^n,X^n}W^n , P_{\Theta,Z^n,X^n} Q_{Y^n})  \nonumber \\
		& \ \ \ = Q\left[ \log \frac{P_{\Theta, Z^n,X^n}W^n}{P_{\Theta,Z^n,X^n} Q_{Y^n}} \ge \log \gamma \right],
	\end{align}
	we can find at least one $\pi \in \Theta, \ x^n \in \mathcal{X}^n$ that satisfies (\ref{eq:meta_converse_proof_1})-(\ref{eq:meta_converse_proof_2}),
			\begin{figure*}[t]
							\normalsize
							\vspace*{4pt}
						   \begin{align}
						   			&Q\left[ \log \frac{P_{\Theta, Z^n,X^n}W^n}{P_{\Theta,Z^n,X^n} Q_{Y^n}} \ge \log \gamma \right] \label{eq:meta_converse_proof_1} \\
						   			& = \sum_{\pi \in \Theta} P_{\Theta}(\pi) \sum_{x^n \in \mathcal{X}^n} \prod_{i=1}^{n}\pi(x_i) \sum_{x^n \in \mathcal{X}^n} P_{X^n|Z^n}(x^n|z^n) \sum_{y^n \in \mathcal{Y}^n} Q_{Y^n}(y^n) \cdot 1 \left\{  \log \frac{W^n(y^n|x^n)}{Q_{Y^n}(y^n)} \ge \log \gamma \right\} \\
						   			& = \sum_{\pi \in \Theta} P_{\Theta}(\pi) \sum_{x^n \in \mathcal{X}^n} \prod_{i=1}^{n}\pi(x_i) \sum_{y^n \in \mathcal{Y}^n} Q_{Y^n}(y^n) \cdot 1 \left\{  \log \frac{W^n(y^n|x^n)}{Q_{Y^n}(y^n)} \ge \log \gamma \right\} \label{eq:meta_converse_proof_3}\\
						   			& \ge \inf \limits_{\pi \in \Theta} \sum_{x^n \in \mathcal{X}^n} \prod_{i=1}^{n}\pi(x_i) \sum_{y^n \in \mathcal{Y}^n} Q_{Y^n}(y^n) \cdot 1 \left\{  \log \frac{W^n(y^n|x^n)}{Q_{Y^n}(y^n)} \ge \log \gamma \right\} \\
						   			& \ge \inf \limits_{\pi \in \Theta} \inf \limits_{x^n \in \mathcal{X}^n} \sum_{y^n \in \mathcal{Y}^n} Q_{Y^n}(y^n) \cdot 1 \left\{  \log \frac{W^n(y^n|x^n)}{Q_{Y^n}(y^n)} \ge \log \gamma \right\} \\
						   			& = \inf \limits_{\pi \in \Theta} \inf \limits_{x^n \in \mathcal{X}^n} Q \left[  \log \frac{W_{x^n}}{Q_{Y^n}} \ge \log \gamma \right]\\
						   			& =  \inf \limits_{\pi \in \Theta} \inf \limits_{x^n \in \mathcal{X}^n} Q \left[  \log \frac{\prod_{i=1}^{n} W_{x_i}}{Q_{Y^n}} \ge \log \gamma \right] \label{eq:meta_converse_proof_2}
						   		\end{align}
						   		  \hrulefill
						\end{figure*}
		where (\ref{eq:meta_converse_proof_3}) follows from the lemma \ref{lemma:permutation_do_not_change_distribution}, and the last equation comes from the stationary memoryless property of DMC (\ref{eq:memroyless_property}). 
		
		Now, let $Q_{Y^n} = \hat{Q}_{Y^n}$, i.e., 
		\begin{equation}
		Q_{Y^n} (y^n) =   \frac{1}{E} \sum_{\mathcal{N}_e \in \mathcal{O}} \prod_{Q_x \in \mathcal{N}_e} \prod_{i=1}^{n\hat{P}_{x^n}(x)} Q_{x}(y_i).
		\end{equation}
	Notice that
	\begin{equation}
		\beta_\alpha(W_{x^n},Q_{Y^n}) = Q \left[  \log \frac{\prod_{i=1}^{n} W_{x_i}}{Q_{Y^n}} \ge \log \gamma \right],
	\end{equation}
	substitute these into (\ref{eq:message_and_beta}) we get 
	\begin{equation}
		\log M^*(n,\epsilon) \le \sup \limits_{\pi \in \Theta} \sup \limits_{x^n \in \mathcal{X}^n} - \log \beta_\alpha(W_{x^n},Q_{Y^n}).
	\end{equation}
	This concludes the proof.
	{\hfill $\square$}

\section{Proof of Lemma \ref{lemma:prob_change}}\label{Appendix B}	
\begin{proof}
	Since elements of $\mathcal{N}_e$ are selected from $\mathcal{N}_r$, we obtain that for any $x \in \mathcal{X}$, there exists a set $\mathcal{N}_e$ such that $Q_{x}^* \in \mathcal{N}_e$ as the divergence covering center of $W_x$. Then we have
					\begin{align}
						&\log \frac{\prod_{i=1}^{n} W(y_i|x_i)}{Q_{Y^n}(y^n)} \nonumber \\ 
						= & \log E + \log\left(   \frac{\prod_{i=1}^{n}W(y_i|x_i)}{\sum_{\mathcal{N}_e \in \mathcal{O}} \prod_{Q_x \in \mathcal{N}_e} \prod_{i=1}^{n\hat{P}_{x^n}(x)} Q_{x}(y_i)} \right) \nonumber\\
						\le & \log E + \log\left(   \frac{\prod_{i=1}^{n}W(y_i|x_i)}{ \prod_{Q_x \in \mathcal{N}_e} \prod_{i=1}^{n\hat{P}_{x^n}(x)} Q_{x}^*(y_i)} \right) \nonumber\\
						= & \log E + \log \left(\frac{\prod_{x \in \mathcal{X}} \prod_{i=1}^{n\hat{P}_{x^n}(x)}W(y_i|x)}{ \prod_{Q_x \in \mathcal{N}_e} \prod_{i=1}^{n\hat{P}_{x^n}(x)} Q_{x}^*(y_i) } \right)  \label{eq:prob_change_proof_1} \\
						= & \log E + \sum_{i=1}^{n}\log  \frac{W(y_i|x_i)}{ Q_{x_i}^*(y_i) } \label{eq:prob_change_proof_2},
					\end{align}
					where (\ref{eq:prob_change_proof_1}) follows from that there are $n\hat{P}_{x^n}(x)$ occurrences of $W_x$. (\ref{eq:prob_change_proof_2}) follows from that $x_i \in \mathcal{X}$. We get $\log \frac{\prod_{i=1}^{n} W(y_i|x_i)}{Q_{Y^n}(y^n)} \ge \log \gamma$ implies 
					\begin{equation}
						\log E + \sum_{i=1}^{n}\log  \frac{W(y_i|x_i)}{ Q_{x_i}^*(y_i) }  \ge \log \gamma.
					\end{equation}
					Then we have
					\begin{align}
						&\mathbb{P}  \left[  \log \frac{W_{x^n}}{Q_{Y^n}} \ge \log \gamma \right] \nonumber \\
						& = \sum_{y^n \in \mathcal{Y}^n} P^n(y^n) 
												\cdot 1\left\{ \log \frac{\prod_{i=1}^{n}W(y_i|x_i)}{Q_{Y^n}} \ge \log \gamma \right\}\nonumber \\
						& \le \sum_{y^n \in \mathcal{Y}^n} P^n(y^n) \cdot1\left\{ \log E + \sum_{i=1}^{n} \log \frac{W(y_i|x_i)}{Q_{x_i}^*(y_i)} \ge \log \gamma \right\}\nonumber \\
						& = \mathbb{P} \left[ \sum_{i=1}^{n} \log \frac{W_{x_i}}{Q_{x_i}^*} \ge \log \gamma -  \log E\right].
					\end{align}
					This completes the proof.
\end{proof}

\section{Proof of Lemma \ref{lemma:asymptotic_bound_of_beta}} \label{Appendix C}
\begin{proof}
	By (\ref{eq:strong_converse_of_NP_region}) we get the converse of the $\beta$ function
			\begin{align}\label{eq:asymptotic_bound_of_beta_1}
				&\beta_{\alpha}(W_{x^n} ,Q_{Y^n}) \nonumber\\
				& \ \ \ \ge \frac{1}{\gamma}\left( \alpha - P\left[ \log \frac{W_{x^n}}{Q_{Y^n}} \ge \log \gamma \right] \right).
			\end{align}
			Apply lemma \ref{lemma:prob_change}, we have 
							\begin{align}
								&  P \left[  \log \frac{W_{x^n}}{Q_{Y^n}} \ge \log \gamma \right] \nonumber \\
								& \ \ \ \le P \left[ \sum_{i=1}^{n} \log \frac{W_{x_i}}{Q_{x_i}^*}  + \log E \ge \log \gamma  \right].
							\end{align}	
				Substituting this into (\ref{eq:asymptotic_bound_of_beta_1}), we get
				\begin{align}\label{eq:asymptotic_bound_of_beta_2}
						&\beta_{\alpha}(W_{x^n} ,Q_{Y^n}) \ge   \nonumber\\
						& \ \ \  \frac{1}{\gamma}\left( \alpha - P \left[ \sum_{i=1}^{n} \log \frac{W_{x_i}}{Q_{x_i}^*} \ge \log \gamma - \log E \right] \right).
				\end{align}
				We observe that $\sum_{i=1}^{n} \log \frac{W_{x_i}}{Q_{x_i}^*}$ is the sum of independent random variables. 
				
				In the case of $V_n > 0$, Berry-Esseen bound \cite[Chapter XVI.5]{Feller_book} shows
				\begin{align}
					P\left[\sum_{i=1}^{n} \log \frac{W_{x_i}}{Q_{x_i}^*} - nD_n \ge \lambda \sqrt{nV_n} \right] 
					\le  \frac{6T_n}{\sqrt{nV_n^3}} +  \Phi(-\lambda). \nonumber
				\end{align}
				For large $n$, let $\Delta > 0$ such that
				\begin{equation}
					\Phi(-\lambda) = \alpha - \frac{6T_n}{\sqrt{nV_n^3}} - \frac{\Delta}{\sqrt{n}}
				\end{equation}
				is positive, let
				\begin{align}
					\log \gamma = & \log E + nD_n \nonumber\\
					& - \sqrt{nV_n} \Phi^{-1} \left( \alpha - \frac{6T_n}{\sqrt{nV_n^3}} - \frac{\Delta}{\sqrt{n}} \right).
				\end{align}
				By Berry-Esseen bounds we have
				\begin{align}
					&P \left[ \sum_{i=1}^{n} \log \frac{W_{x_i}}{Q_{x_i}^*} \ge \log \gamma - \log E \right] \le \alpha -\frac{\Delta }{\sqrt{n}}. 
				\end{align}
				Substituting this into (\ref{eq:asymptotic_bound_of_beta_2}), we get
				\begin{align}
					&\log \beta_{\alpha}(W_{x^n} , Q_{Y^n}) \nonumber\\
					&\ \ \ \ge -\log E - nD_n + \log \frac{\Delta}{\sqrt{n}} \nonumber\\
					& \ \ \ \ \ \ + \sqrt{nV_n} \Phi^{-1}\left( \alpha - \frac{6T_n}{\sqrt{nV_n^3}} - \frac{\Delta}{\sqrt{n}}\right). 
				\end{align}
				
				In any case, we let
				\begin{equation}
					\log \gamma = \log E + nD_n + \sqrt{\frac{2nV_n}{\alpha}},  
				\end{equation}
				by Chebyshev inequality, we have 
				\begin{equation}
					P\left[  \sum_{i=1}^{n} \log \frac{W_{x_i}}{Q_{x_i}^*} \ge \log \gamma - \log E \right] \le \frac{\alpha}{2}.
				\end{equation}
				Again by (\ref{eq:asymptotic_bound_of_beta_2}) we get 
				\begin{align}
					\log \beta_{\alpha}(W_{x^n} , Q_{Y^n}) & \ge -\log E - nD_n \nonumber\\ 
					& \ \ \ - \sqrt{\frac{2nV_n}{1 - \epsilon}} + \log \frac{\alpha}{2}. 
				\end{align}		
\end{proof}

\section{Proof of Lemma \ref{lemma:properities_Vn_Tn}} \label{Appendix D}

\begin{proof}
			We use (\ref{eq:construst_covering_center_1})-(\ref{eq:construst_covering_center_3}) to construct a grid-like structure. Then, there exists a constant $\tau$ such that for $r_0 = \frac{1}{n\tau}$, the $\Lambda_\ell\left(\frac{1}{n\tau}\right)$ is the set of covering centers of $\Delta_{ \ell}$, and the covering radius is $1/n$. Let $\mathcal{N}_r = \Lambda_\ell\left(\frac{1}{n\tau}\right)$, we first consider the $2$-dimensional case. Let $W(y|x) = \delta \le 1/2$ and $r_0 (i-1)^2 \le \delta \le r_0 i^2$. 
			
			\textit{Case(a)}. $W(y|x) \ge Q_{x}^*(y)$: Let $Q_{x}^*(y) = r_0 (i-1)^2 $, note that $\sqrt{\delta \tau n} \le i \le \sqrt{\delta \tau n} + 1$, we have
			\begin{align}
				\left| \log \frac{W(y|x)}{Q_{x}^*(y)} \right| 
				& \le \left( \frac{W(y|x)}{Q_{x}^*(y) }-1 \right) \log e  \nonumber \\
				& \le \frac{r_0i^2 - r_0(i-1)^2}{r_0 (i-1)^2}  \log e \nonumber \\
				& \le \frac{2 - 1/\sqrt{\delta \tau n}}{\sqrt{\delta \tau} \sqrt{n} - 2 + 1/\sqrt{\delta \tau n}} \log e.
			\end{align}
			Then, for $n$ sufficiently large, we can find a constant $F_y$ only depend on $y \in \mathcal{Y}$ such that
			\begin{align}
				\left| \log \frac{W(y|x)}{Q_{x}^*(y)} \right| \le \frac{F_y}{\sqrt{n}}
			\end{align}
			
			\textit{Case(b)}. $Q_{x}^*(y) \ge W(y|x)$: Let $Q_{x}^*(y) = r_0 i^2 $, note that $W(y|x) \ge r_0(i-1)^2$, we can apply the same argument in \textit{Case(a)} to show
			\begin{align}
				\left| \log \frac{W(y|x)}{Q_{x}^*(y)} \right|  
				& \le \left( \frac{Q_{x}^*(y)}{W(y|x)} -1 \right) \log e \nonumber \\
				& \le \frac{r_0i^2 - r_0(i-1)^2}{r_0 (i-1)^2} \log e  \nonumber \\
				& \le \frac{F_y}{\sqrt{n}}
			\end{align}
			
			If $\delta > 1/2$, by symmetric we can get the same result in \textit{Case(a)} and \textit{Case(b)}.
			
			In $|\mathcal{Y}|$-dimensional case, let $C(y)$ denote a constant only depend on $y \in \{ 1,...,|\mathcal{Y}| \}$. Then we have $  \frac{C(y)}{n} i^2 \le \delta \le \frac{C(y)}{n} (i+1)^2$ because each $k$-th dimension is constructed by $\Lambda_2\left( \frac{1}{n C_0(y)}\right)$, where $C_0(y)$ is a constant only depend on $y$. Following the same argument in the $2$-dimensional case, we get the conclusion.
			
			Note that $D(W_x\| Q_{x}^*)^2 \le \frac{1}{n^2}$, we get
			\begin{equation}
				V_n \le \frac{C_1}{n},
			\end{equation}
			where $C_1$ is a constant. Then, using the inequality $|a-b|^3 \le 4(|a|^3 - |b|^3)$ to obtain
			\begin{equation}
				T_n \le \frac{C_2}{n^{3/2}},
			\end{equation}
			where $C_2$ is a constant. This completes the proof.
\end{proof}
	
\section{Proof of Theorem \ref{thm:asymptotic_bound_of_permutation_channel} }  \label{Appendix E}

\begin{proof}
	Let the covering radius $r = 1/n$, apply (\ref{eq:covering_number_converse}), we obtain 
				\begin{equation}
					\log |\mathcal{N}_r| \le \frac{\ell}{2} \log n + \log C_1,
				\end{equation}
				where $\log C_1 = \frac{\ell}{2}( \log \ell + 2 \log (\sqrt{2\tau}+1) )$ and $\ell = |\mathcal{Y}|-1$. Note that $E \le |\mathcal{N}_r|/|\mathcal{X}| + 1$, for any $n \ge 1$, we can find a constant $C_2$ satisfies
				\begin{equation}\label{eq:upper_bound_of_E}
					\log E \le \frac{\ell}{2}\log n + \log C_1 + \log C_2 - \log |\mathcal{X}|				
				\end{equation}
							
					\textit{Case(a)}. $\epsilon \in \left(0,1/2\right]$. Note that we have
					\begin{align}
						V_n & = \frac{1}{n} \sum_{i=1}^{n}\mathbb{E}_{W_{x_i}}\left[\left(\log \frac{W_{x_i}}{Q_{x_i}^*} - D(W_{x_i}\|Q_{x_i}^*)\right)^2 \right] \nonumber \\
						& = \mathbb{E}_{\hat{P}_{x^n} W}\left[  \left( \log \frac{W}{Q_{X}^*} - D(W\|Q_{X}^*)\right)^2\right] \nonumber \\
						& = V(\hat{P}_{x^n},W).
					\end{align}
					Apply Lemma \ref{lemma:properities_Vn_Tn} we obtain that for any $\pi \in \Theta$, there exists a constant $F_0$ only depends on $\pi$ and $\hat{P}_{x^n}$ such that
					\begin{equation}
						|V(\hat{P}_{x^n},W) - V(\pi,W)| \le \frac{F_1}{n}.
					\end{equation}
					Then, apply Proposition \ref{proposition:meta_converse_of_permutation channel} and Lemma \ref{lemma:asymptotic_bound_of_beta} to get
					\begin{align} \label{eq:upper_bound_of_M}
						\log M 
						& \le  \log E + nD_n - \log \frac{\Delta}{\sqrt{n}} + O(1)\nonumber \\
						&  \ \ \ - \sqrt{nV(\pi,W)} \Phi^{-1} \left(  \alpha - \frac{T_n}{\sqrt{nV_n^3}} - \frac{\Delta}{\sqrt{n}}\right) .
					\end{align}
					Let
					\begin{equation}
						\Delta = 2^{-G_1 \log \log n} \times \sqrt{n},
					\end{equation}
					where $G_1$ is a constant. Then, we notice that $ \log \frac{\Delta}{\sqrt{n}} = -G_1\log \log n$ and $\frac{\Delta}{\sqrt{n}} = 2^{- G_1 \log \log n}$. For any $F_2\in \left(0,1\right)$, we can find $n \ge N_0$ such that $\frac{\Delta}{\sqrt{n}} \le F_2$. Again by Lemma \ref{lemma:properities_Vn_Tn}, for any $F_3 \in (0,1)$, we can find $n \ge N_1$ such that $6 T_n/\sqrt{nV_{n}^3}\le F_3$. Note that $ \sqrt{nV(\pi,W)} $ is a finite constant for any $n$, setting $n \ge \max \{ N_0,N_1 \}$ and applying Taylor's formula we have
					\begin{align}
						  \sqrt{nV(\pi,W)} \ & \Phi^{-1}  \left(  \alpha - \frac{6 T_n}{\sqrt{nV_{n}^3}} - \frac{\Delta}{\sqrt{n}}\right) \nonumber \\
						 &  \ \ \ge \sqrt{nV(\pi,W)} \Phi^{-1}(\alpha) - F_4 
					\end{align}
					for some finite constant $F_4$ and all $n \ge \max \{ N_0,N_1 \} $.
					Substituting these into (\ref{eq:upper_bound_of_M}), and note that $\Phi^{-1}(\alpha) = - \Phi^{-1}(\epsilon)$ we have
					\begin{align}
						\log M \le& \  \frac{ \ell}{2} \log n  + \sqrt{nV_{min}} \Phi^{-1}(\epsilon) \nonumber \\
						&+O(\log \log n).
					\end{align}
					
					\textit{Case(b)}. $\epsilon \in  \left( 1/2, 1 \right)$. Following the same argument in \textit{Case(a)} but maximizing $V(\pi,W)$ because $\Phi(\epsilon) \ge 0$, then we get 
					\begin{align}
						\log M \le  & \ \frac{ \ell}{2} \log n  + \sqrt{nV_{max}} \Phi^{-1}(\epsilon) \nonumber \\
						&+O(\log \log n).
					\end{align}
					This completes the proof.
\end{proof}

\begin{figure*}[t]
					\normalsize
					\vspace*{4pt}
						\begin{align} 
							  \Lambda \left( r_0 \right) =& \left\{ r_0 i^2 : i \in \mathbb{Z}, r_0 i^2< \frac{1}{2} \right\} \cup \left\{  1-r_0 i^2 : i \in \mathbb{Z}, 1-r_0 i^2 < \frac{1}{2} \right\} \cup \left\{ \frac{1}{2} \right\} \label{eq:construst_covering_center_1}\\
							  \Lambda_2\left( r_0 \right) =& \left\{  (q,1-q):q \in \Lambda \left( r_0 \right)\right\} \label{eq:construst_covering_center_2}\\
							  \Lambda_k(r_0) =& \bigcup \limits_{ q \in \Lambda( r/k)} \left\{ ((1-q)q_1^*,...,(1-q)q_{k-1}^*,q): (q_1^*,...,q_{k-1}^*) \in \Lambda_{k-1}\left( \frac{k-1}{k} r_0 \right) \right\} \label{eq:construst_covering_center_3}
						\end{align}
					\hrulefill
\end{figure*}
\section{Proof of Theorem \ref{thm:asymptotic_bound_of_BSC_permutation_channel}} \label{Appendix F}	
\begin{proof}
	We use (\ref{eq:construst_covering_center_1}) and (\ref{eq:construst_covering_center_2}) to construct a set of points $\Lambda_2(r_0)$. Theorem \ref{thm:meta_converse_of_BSC_permutation_channel} yields that we only need to focus on the point within $(\delta,1-\delta)$. To do this we define
			\begin{equation}
				\Delta_{1}^* = \{ (p,1-p) \in \mathbb{R}^2, \delta < p \le 1/2 \}.
			\end{equation}
			Then we need to cover the $\Delta_{1}^*$.
			Without loss of generality, let $\delta < 1/2$, we can set $ r_0 (i-1)^2 < \delta \le r_0 i^2 < 1/2$ for small $r_0$. Let $Q_1 = (r_0i^2,1-r_0i^2)$, $Q_2 = (r_0(i-1)^2,1-r_0(i-1)^2)$, from (\ref{eq:KL_is_upperbounded_by_Chi})  we obtain that
			\begin{equation}\label{BSC converse use 5}
				D(W_x\|Q_1) \le \frac{(\delta - r_0 i^2)^2}{r_0i^2 (1- r_0i^2) }
			\end{equation}
			and
			\begin{equation}\label{BSC converse use 6}
				D(W_x\|Q_2) \le \frac{(\delta - r_0 (i-1)^2)^2}{r_0(i-1)^2 (1- r_0(i-1)^2) }.
			\end{equation}
			
			\textit{Case(a).} When $\delta \ge r_0 (i-\sqrt{2}/2)^2$, let $Q_1$ be the covering center of $W_x$ and apply (\ref{BSC converse use 5}) to obtain 
			\begin{align}
				D(W_x \| Q_1) & \le \frac{(r_0(i-\sqrt{2}/2)^2 - r_0i^2)^2}{r_0i^2(1-r_0i^2)} \nonumber\\
																& \le \frac{r_0}{\delta} \times \frac{(-\sqrt{2}i +0.5)^2}{1/r_0 - i^2},
			\end{align}
			where the second inequality follows from the fact that $r_0i^2 > \delta$. 
			For $i < \sqrt{1/r_0}$, we can find a $c > 0$ such that
			\begin{equation} \label{BSC converse use7}
				D(W_x \| Q_1) \le \frac{r_0c}{\delta}.
			\end{equation}
			
			\textit{Case(b).} When $\delta < r_0 (i-\sqrt{2}/2)^2$, let $Q_2$ be the covering center of $W_x$ and apply (\ref{BSC converse use 6}) to obtain
			\begin{align} 
				D(W_x \| Q_2) & \le \frac{(r_0(i-\sqrt{2}/2)^2 - r_0(i-1)^2)^2}{r_0(i-1)^2 (1- r_0(i-1)^2)} \nonumber\\
				& = \frac{r_0^2((2-\sqrt{2})i-1/2)^2}{r_0(i-1)^2 (1- r_0(i-1)^2)} \nonumber \\
				& \le r_0 \times \frac{r_0^2((2-\sqrt{2})i-1/2)^2}{r_0(i-1)^2/2},
			\end{align}
			where the last inequality follows from the fact that $1- r_0(i-1)^2 > 1/2$. For $i \ge 2$, we have
			\begin{equation} \label{BSC converse use8}
				D(W_x \| Q_2) \le  r_0.
			\end{equation}
				
			Let $\tau = c/\delta$ and $ r_0 =  \frac{1}{\tau n}$, where we set $c > 1/2$. 
			Note that we can find $i$ such that
			\begin{equation}\label{BSC converse use2}
				\lambda_2 = \frac{1}{\tau n} (i-1)^2  \le \delta \le \frac{1}{\tau n} i^2 = \lambda_1.
			\end{equation}
			Therefore we have 
			\begin{equation}\label{BSC converse use1}
				\sqrt{\delta \tau n } \le i \le \sqrt{\delta \tau n } +1.
			\end{equation}
			Note that for $n$ sufficiently large, we have $i \ge 2$ and $i \le \sqrt{\tau n \delta} +1 \le \sqrt{1/r_0}$. According to (\ref{BSC converse use7}) and (\ref{BSC converse use8}) we get
			\begin{equation}
				 \min \limits_{ Q \in \Lambda_2(r_0)}D(W_x \| Q) \le \frac{1}{n}.
			\end{equation}
			For the point $P = (p,1-p)$, where $p > \delta$, by the same argument we get 
			\begin{equation}
				\min \limits_{ Q \in \Lambda_2(r_0)}D(P \| Q) \le \frac{1}{n}.
			\end{equation}
			That is, all points of $\Delta_{1}^*$ are covered by covering centers and the covering radius is $1/n$. 
			
			Then we need to upper bound $c$, for $r_0 = \frac{1}{\tau n}$ and $i \le \sqrt{\tau n \delta} +1 $, we have
			\begin{align}
				c &=  \frac{(-\sqrt{2}i +0.5)^2}{\frac{cn}{\delta} - i^2} \nonumber\\
				& \le \frac{(-\sqrt{2cn} + 0.5)^2}{cn(1-\delta)/\delta  - 2\sqrt{cn} - 1}.
			\end{align}		
			Because $\delta < 1/2$, we have $(1-\delta) /\delta > 1$. Then, for $n$ sufficiently large, we have
			\begin{equation}
				\frac{(-\sqrt{2cn} + 0.5)^2}{cn(1-\delta)/\delta  - 2\sqrt{cn} - 1} \le 2.
			\end{equation}
			Thus, we set $c=2$.
			
			Now, we compute the variance. Note that for any $\delta > r_0 (i-\sqrt{2}/2)^2$, $ Q_1 = (\lambda_1,1-\lambda_1)$ is always the divergence covering center of $W_x$. Thus we can let
			\begin{align} \label{BSC converse use4}
				\lambda_1 - \delta = \frac{1}{2} \times \frac{1}{\tau n} \left[i^2 - (i-1)^2\right].
			\end{align}
			Substituting this into (\ref{BSC converse use1}) we get
			\begin{equation}
				\delta +  \sqrt{\frac{\delta}{\tau n}} - \frac{1}{2\tau n} \le \lambda_1 \le \delta + \sqrt{\frac{\delta}{\tau n}} + \frac{1}{2\tau n}.
			\end{equation}
			Then the lower bound and upper bound of the variance can be computed as (\ref{eq:lower_bound_of_variance}) and (\ref{eq:upper_bound_of_variance}). 
			\begin{figure*}[t]
						\normalsize
						\vspace*{4pt}
						\begin{align}\label{eq:lower_bound_of_variance}
								V_{min}(\delta) & =
								\delta \log ^2 \frac{\delta}{\delta + \sqrt{\frac{\delta}{\tau n}} - \frac{1}{2\tau n}} 
								+ (1- \delta) \log^2 \frac{1-\delta}{1 - \delta - \sqrt{\frac{\delta}{\tau n}} + \frac{1}{2\tau n}} \nonumber\\
								& \ \ \ - \left( \delta \log  \frac{\delta}{\delta + \sqrt{\frac{\delta}{\tau n}} + \frac{1}{2\tau n}} 
								+ (1- \delta) \log \frac{1-\delta}{1 - \delta -\sqrt{ \frac{\delta}{\tau n}} - \frac{1}{2\tau n}}\right)^2
						\end{align}
						\begin{align}\label{eq:upper_bound_of_variance}
								V_{max}(\delta) & =
								\delta \log ^2 \frac{\delta}{\delta + \sqrt{\frac{\delta}{\tau n}} + \frac{1}{2\tau n}} 
								+ (1- \delta) \log^2 \frac{1-\delta}{1 - \delta - \sqrt{\frac{\delta}{\tau n}} - \frac{1}{2\tau n}} \nonumber\\
								&\ \ \ - \left( \delta \log  \frac{\delta}{\delta + \sqrt{\frac{\delta}{\tau n}} - \frac{1}{2\tau n}} 
								+ (1- \delta) \log \frac{1-\delta}{1 - \delta -\sqrt{ \frac{\delta}{\tau n}} + \frac{1}{2\tau n}}\right)^2
						\end{align}
						\hrulefill
			\end{figure*}		
			
			Since Theorem \ref{thm:asymptotic_bound_of_permutation_channel} holds for the set of covering centers which constructed in any way, for $\epsilon \in (0,1/2)$ we have
			\begin{align}
				\log M \le & \ \frac{1}{2} \log n + \sqrt{nV_{min}(\delta)} \Phi^{-1}(\epsilon) \nonumber \\ 
				&+O(\log \log n), 
			\end{align}
			and for $\epsilon \in \left[1/2,1\right)$ we have
			\begin{align}
				\log M \le & \ \frac{1}{2} \log n + \sqrt{nV_{max}(\delta)} \Phi^{-1}(\epsilon) \nonumber \\ 
				&+O(\log \log n).
			\end{align}
			This completes the proof of (\ref{eq:asymptotic_bound_of_BSC_permutation_channel_1}) and (\ref{eq:asymptotic_bound_of_BSC_permutation_channel_2}).
\end{proof}

\printbibliography[heading=bibintoc, title={References}]
		
		\nocite{shiryayev_selected_1993}
		\nocite{yang_information-theoretic_1999}

		\newpage

\end{document}